\documentclass[11pt]{amsart}
\usepackage{rotating}
\usepackage{amsthm}
 \usepackage{amsmath}
\usepackage{graphics}
 \usepackage{amsfonts}
 \usepackage{amssymb}
 \usepackage{amscd}
 \usepackage[all]{xy}
\usepackage{colordvi}
\usepackage{color}
\usepackage{pdfsync}
\usepackage{url}
\usepackage{enumerate}
\usepackage{booktabs}
\usepackage{multicol}
\usepackage[english]{babel}
\usepackage{blindtext}
\usepackage{subfigure} 
\usepackage{hyperref, enumerate} 

\newtheorem{thm}{Theorem}[section]

\newtheorem{prop}[thm]{Proposition}
\theoremstyle{definition}
\newtheorem{ex}[thm]{Example}

\theoremstyle{definition}
\newtheorem{definition}[thm]{Definition}
\newtheorem{rem}[thm]{Remark}
\numberwithin{equation}{section}

\def \R { {\mathbb R} }

\def \E { {\mathbb E} }

  \DeclareMathOperator{\conv}{conv}

\title{Sequential Importance Sampling for Two-dimensional Ising Models}
\author{
Jing Xi \and Seth Sullivant}
\address{Department of Mathematics, North Carolina State University, Raleigh, NC, USA } 
\email{jxi2@ncsu.edu} 
\email{smsulli2@ncsu.edu } 

\date{}

\begin{document}

\maketitle                 

\begin{abstract}
In recent years, sequential importance sampling (SIS) has been well developed for sampling contingency tables with linear constraints. In this paper, we apply SIS procedure to 2-dimensional Ising models, which give observations of 0-1 tables and include both linear and quadratic constraints. We show how to compute bounds for specific cells by solving linear programming (LP) problems over cut polytopes to reduce rejections. The computational results, which includes both simulations and real data analysis, suggest that our method performs very well for sparse tables and when the 1's are spread out: the computational times are short, the acceptance rates are high, and if proper tests are used then in most cases our conclusions are theoretically reasonable.

\bigskip

\noindent {\bf Keywords:} sequential importance sampling, ising model, cut polytope

\end{abstract}

\section{Introduction}\label{sec:intro}

The Ising models, which were first defined by Lenz in 1920, find its applications in physics, neuroscience, agriculture and spatial statistics \cite{lee1952, hopfield1982, besag1978}. In 1989, Besag and Clifford showed how to carry out exact tests for the 2-dimensional Ising models via Monte Carlo Markov chain (MCMC) sampling \cite{besag1989}. However, their approach does not lead to a connected Markov chain, i.e. their Markov moves cannot connect all tables. Martin del Campo, Cepeda and Uhler developed a MCMC strategy for the Ising model that avoids computing a Markov basis. Instead, they build a Markov chain consisting only of swaps of two interior sites while allowing a bounded change of sufficient statistics \cite{abraham}.

Sequential importance sampling (SIS) is a sampling method 
which can be used to sample contingency tables with linear 
constraints \cite{chen2005}. It proceeds by sampling the table 
cell by cell with specific pre-determined marginal distributions, 
and terminates at the last cell. It was shown that compared 
with MCMC-based approach, SIS procedure does not require 
expensive or prohibitive pre-computations, and when there is 
no rejection, it is guaranteed to sample tables i.i.d.~from the proposal distribution, while in an MCMC approach the computational 
time needed for the chain to satisfy the independent condition is unknown.

In this paper, we describe an SIS procedure for the $2$-dimensional
Ising model.
We first define the 2-dimensional Ising models in Section \ref{sec:ising}, 
and give a brief review of SIS method in Section \ref{sec:sis}. In Section 
\ref{sec:cutpoly} we describe how to define cut polytopes for Ising models 
and how to compute bounds by solving LP problems on the cut polytope. 
Then in Section \ref{sec:comp} we give computational results from 
simulations of both Ising models and autologistic regression models, 
and real dataset analysis. Lastly, we end the paper with discussion and 
future questions on the SIS procedure for Ising models.


\section{Two-dimensional Ising models}\label{sec:ising}

Suppose $L$ is a grid (or table) with size $m\times n$, and a zero-one random variable $X_{i,j}$ is assigned as the $(i,j)_{th}$ entry of $L$. Then $X=(X_{1,1}, X_{1,2}, \ldots, X_{m,n})\in\{0,1\}^{mn}$ is a configuration of $L$. 
For sites $(i_1, j_1)$ and $(i_2, j_2)$ in $L$, we define $(i_1, j_1) \sim (i_2, j_2)$ if and only if $(i_1, j_1)$ and $(i_2, j_2)$ are nearest neighbors in the grid $L$.
The Ising model \cite[\S 4]{besag1989}  is defined as the 
set of all probability distributions on $\{0,1\}^{mn}$ of the form:
\begin{equation}\label{eq:ising}
Pr(X=x) = \frac{\exp (\alpha T_1(x) + \beta T_2(x))}{C(\alpha, \beta)},
\end{equation}
for some $\alpha, \beta \in \R$,
where $T_1(x)=\sum\limits_{(i,j)\in L} x_{i,j}$, $T_2(x)=\sum\limits_{(i_1, j_1) \sim (i_2, j_2)} (x_{i_1, j_1}(1-x_{i_2, j_2}) + x_{i_2, j_2}(1-x_{i_1, j_1}))$ and $C(\alpha, \beta)$ is the appropriate normalizing constant. Notice that $(T_1(X), T_2(X))$ are sufficient statistics of the model,
where $T_{1}(X)$ is the number of ones in the table,
and $T_{2}(X)$ is the number of adjacent pairs that are not equal.

For an observed table $x_0$ with $T_1=T_1(x_0)$ and $T_2=T_2(x_0)$, we can define the fiber as:
\begin{equation}\label{eq:suff}
\Sigma = \{x \in\{0,1\}^{mn} | T_1(x)=T_1,\ T_2(x)=T_2\}.
\end{equation}
The conditional distribution of $L$ given $(T_1, T_2)$ is the uniform distribution over $\Sigma$.

In \cite{besag1989}, Besag and Clifford were interested performing a hypothesis test
to see if there exist some unspecified $\alpha$ and $\beta$ such that the observed table $x_0$ is a realization of the Ising model defined in Equation \ref{eq:ising}. 
The test statistic they used was $u(x_0)$, the number of diagonally adjacent $1$'s in $L$, i.e. the number of $2\times 2$ windows of forms
$\begin{array}{|c|c|}                                                                                       
\hline                                                                                                       
1&0\\                                                                                                      
\hline                                                                                                       
0&1 \\                                                                                                     
\hline                                                                                                       
\end{array}$ or
$\begin{array}{|c|c|}                                                                                       
\hline                                                                                                       
0&1\\                                                                                                      
\hline                                                                                                       
1&0 \\                                                                                                     
\hline                                                                                                       
\end{array}$ in $L$.
To carry out the test, they estimated $\alpha$ and $\beta$ by maximum pseudo-likelihood and construct a Monte Carlo Markov chain to sample a series of tables $x_1, x_2, \ldots, x_N$. 
They estimate a one-sided p-value which is the percentage of $x_i$'s such that $u(x_i) > u(x_0)$ and the percentage of $x_i$'s such that $u(x_i) \geq u(x_0)$  and reject the Ising model if both of these two percentages were significant.


\section{Sequential importance sampling (SIS)}\label{sec:sis}

Let $p(X) = \frac{1}{|\Sigma|}$ be the uniform distribution over 
$\Sigma$, and $q(X)$ be a  distribution such that $q(X) >0$, $\forall\ X\in\Sigma$
that is ``easy'' to sample from directly, and such that 
the probability $q(X)$ is easy to calculate.
In importance sampling, 
 the distribution $q(\cdot)$ serves as a proposal distribution 
which we use to calculate p-values of the exact test.

Consider the conditional expected value $\mathbb E_{p}[f(X)]$ for some statistic $f(X)$, we have:
\[
\mathbb E_p[f(x)] = \int f(x)p(x)dx = \int [f(x) \frac{p(x)}{q(x)}]q(x) dx = \mathbb E_q[f(x)\frac{p(x)}{q(x)}].
\]
Since $p(\cdot)$ is known up to a constant $|\Sigma|$, i.e. $p(X) \propto p^*(X) = 1$, we also have:
\[
\mathbb E_p[f(x)] = \frac{\int f(x) p^*(x)dx}{\int p^*(x) dx} 
= \frac{\int f(x) \frac{p^*(x)}{q(x)} q(x)dx}{\int \frac{p^*(x)}{q(x)} q(x) dx}
= \mathbb E_q\left[f(x) \frac{\frac{p^*(x)}{q(x)}}{\int \frac{p^*(x)}{q(x)} q(x) dx}\right].
\]
Thus we can give an unbiased estimator of $\mathbb E_p[f(x)]$:
\begin{equation}\label{eq:imptest}
\hat{f_n} = \frac{\frac{1}{n}\sum\limits_{i=1}^n f(X_i)\frac{p^*(X_i)}{q(X_i)}}{\frac{1}{n}\sum\limits_{i=1}^n \frac{p^*(X_i)}{q(X_i)}} = \sum\limits_{i=1}^n f(X_i)\omega_s (X_i),
\end{equation}
where ${ X_1}, \ldots , { X_N}$ are tables drawn i.i.d. from $q({X})$, and $\omega_s (X_i)=\frac{\frac{p^*(X_i)}{q(X_i)}}{\sum\limits_{i=1}^n \frac{p^*(X_i)}{q(X_i)}}=\frac{\frac{1}{q(X_i)}}{\sum\limits_{i=1}^n \frac{1}{q(X_i)}}$ are called the standardized importance weights 
\cite{givens2005}.
In this case, the p-value Besag and Clifford suggested in \cite{besag1989} was the p-value for an exact test based on $u(X)$ (which is the same as a volume test \cite{diaconis1983}), is between $p_1=\mathbb E_{p}[\mathbb{I}_{u( X) > u({x_0})}]$ and $p_2=\mathbb E_{p}[\mathbb{I}_{u( X) \geq u({x_0})}]$, where $\mathbb{I}$ is the indicator function. Therefore, we need to estimate the values $p_1$ and $p_2$ by sampling tables over $\Sigma$ using the easy distribution $q$ and the importances
weights.

In sequential importance sampling (SIS) \cite{chen2005}, we vectorize the table as $x = (x_1, \ldots, x_{mn})$. Then by the multiplication rule 
for conditional probabilities we have
\[
q({X} = x) = q(x_1)q(x_2|x_1)q(x_3|x_2, x_1)\cdots
q(x_{mn}|x_{mn-1}, \ldots , x_1). 
\]
This implies that we can sample each table cell by cell with some pre-specified conditional distribution $q(x_{i} \mid x_{i-1}, \ldots, x_1)$.
Ideally, the distribution
$q(\cdot)$ should be as close to $p(\cdot)$ as possible. 
If the sample space $\Sigma$ is especially complicated,
ee may have rejections when we are sampling tables from a bigger set $\Sigma^*$ such that $\Sigma \subset \Sigma^*$. 
Since $q(x_i|x_{i-1}, \ldots ,x_1)$, $i = 2, 3, \ldots mn$, and $q(x_1)$ are normalized, it is easy to show that $q(X)$ is normalized over $\Sigma^*$, hence the estimators are still unbiased because:
\[
\E_q \left[{\mathbb{I}_{{ X} \in \Sigma}} f(X) \omega_s(X)\right] 
= \sum_{{ X} \in \Sigma^*} \left[{\mathbb{I}_{{ X} \in \Sigma}}  f(X) \frac{\frac{1}{q(X)}}{\int \frac{1}{q(X)} q(X) dx}\right]  q(X)
= \sum_{{ X} \in \Sigma^*}  {\mathbb{I}_{{ X} \in \Sigma}}  f(X) \frac{1} {|\Sigma|}
= \E_p [ f(X) ] ,
\]
where $\mathbb{I}_{{ X} \in \Sigma}$ is an indicator function for $\Sigma$. 

We can see that the proposal distribution $q({X} = x)$ heavily depend on how we define the distributions $q(x_1)$ and  $q(x_i|x_{i-1}, \ldots ,x_1)$, $i = 2, 3, \ldots mn$. However, the naive choice of the marginal distributions, $q(x_i=1|x_{i-1}, \ldots ,x_1) = \frac{T_1 - \sum_{j=1}^{i-1} x_j}{N-(i-1)}$, $i = 2, 3, \ldots mn$, and $q(x_1=1) = \frac{T_1}{N}$ works surprisingly bad: the rejection rate is very high even for small cases. Since in our case there is not a standard way to construct the marginal distributions, we build up the distributions based on $T_1$, $T_2$ and the neighbors by some combinatorial justifications. We will show in Section \ref{sec:comp} and \ref{sec:diss} that compare with the naive ones, our marginal distributions are better, but still have the problem of large $cv^2$ which need further work in this aspect.


\section{Linear programming (LP) problem for Ising models on the cut polytope}\label{sec:cutpoly}

As we discussed in the last section, choosing good proposal distributions
$q(x_{i}|x_{i-1}, \ldots ,x_1)$  is a crucial and typical difficulty in the SIS method.  For many models, a natural first step to improving the SIS procedure
is to reduce rejections as much as possible.  
The rejection rate in the simple SIS procedure described in Section \ref{sec:sis} varies according to how much the $1$'s are crowded, which can be measured by $T_2/ T_1$: note that the ratio $T_2/ T_1$ is always between 0 and 4, and it is smaller when $1$'s are more crowded together. Computational results show that we have higher rejection rates when $T_2/ T_1$ is smaller.
 
One straightforward way to reduce rejection is enumerating some special cases
when only a few $1$'s remained to be placed using combinatorial techniques.  
However, only very limited things can be done in this way because of 
the massive number of possible ways of taking values for the table.
In this section, we will show how to use cut polytopes in the 
procedure to reduce the rejection rate: we compute the lower bounds 
and upper bounds for some cells in the table to avoid some cases that will certainly lead to rejections, and this essentially 
narrows the sample space $\Sigma^*$ down to a smaller set which 
contains the fiber $\Sigma$.

Let $G=(V,E)$ be a graph with vertex set $V$ and edge set $E$.
\begin{definition}\label{def:cut}\cite[\S 4.1]{deza1997}
Suppose $A \mid B$ is a partition of $V$. Then the cut semimetric, 
$\delta_{A \mid B}$, is a vector such that for all $ ij \in E$,
\[
\delta_{A\mid B}(i,j) = \begin{cases} 1 & \text{if } \mid \{i,j\}\cap A \mid = 1 \\0 & \text{otherwise} \end{cases}.
\] 
The cut polytope is the convex hull of the cut semimetrics of all
 possible partitions of $V$:
\[
CUT^\Box (G) = \conv\{\delta_{A \mid B}: A\mid B \mbox{ is a partition of } G\}.
\]
\end{definition}

\begin{ex}
Let $G$ be the graph of Figure \ref{fig:cut}. 
Consider a partition of $V$: $A=\{1,2,5\}$ and $B=\{3,4,6\}$, then the cut semimetric is (see Figure \ref{fig:cut}, edges with coordinates $1$ are bolded):
\[
\delta_{A\mid B} = \bordermatrix{ e: e\in E & 12 & 23 & 45 & 56 & 14 & 25 & 36 \cr
& 0 & 1 & 1 & 1 & 1 & 0 & 0  
} \in CUT^\Box (G).
\]
It is also easy to see that $(0,1,1,1,0,0,0)$ is not in $CUT^\Box (G)$,
since the only $0/1$ vectors in $CUT^{\Box}(G)$ are the cut
semimetrics, and any cut semimetric has an even number of $1$'s around
a cycle.
\begin{figure}[!htp]
\begin{center}
\scalebox{1.5}{
\includegraphics{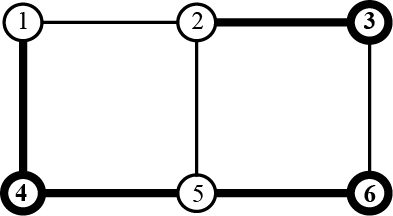}
}
\end{center}
\caption{An example of a graph with $V$=\{1, 2, 3, 4, 5, 6\}.}
\label{fig:cut}
\end{figure}
\end{ex}

Cut polytopes are well-studied objects (see the book \cite{deza1997}
for a great many results about them).  Here are some examples
that we make use of in our calculations related to the
Ising model.

\begin{ex}\cite{deza1997}\label{ex:tri}
The cut polytope of a triangle in Figure \ref{fig:tri} can be defined by the following inequalities:
$0 \leq x_e \leq 1$ for $e=a,b,c$, and:
$$ \begin{array}{rrrrrcc}
x_a & + & x_b & + & x_c & \leq & 2 \\
x_a & + & x_b & - & x_c & \geq & 0 \\
x_a & - & x_b & + & x_c & \geq & 0 \\
-x_a & + & x_b & + & x_c & \geq & 0
\end{array}
$$
\begin{figure}[ht]
  \centering
  \subfigure[A triangle]{
    \label{fig:tri}     
    \includegraphics[width=1.2in]{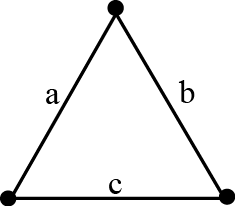}}
  \hspace{0.5in}
  \subfigure[A square]{
    \label{fig:sq}     
    \includegraphics[width=1.2in]{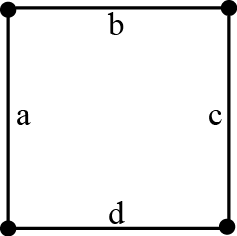}}
  \caption{Graphs to illustrate the inequalities for cut polytopes}
  \label{fig:cutpoly}     
{\footnotesize }
\end{figure}
\end{ex}

\begin{ex}\cite{deza1997}\label{ex:sq}
The cut polytope of a square in Figure \ref{fig:sq} can be defined by the following inequalities:
$0 \leq x_e \leq 1$ for $e=a,b,c,d$, and:
$$ \begin{array}{ccrrrrrrrcc}
0 & \leq & x_a & + & x_b & + & x_c & - & x_d & \leq & 2 \\
0 & \leq & x_a & + & x_b & - & x_c & + & x_d & \leq & 2 \\
0 & \leq & x_a & - & x_b & + & x_c & + & x_d & \leq & 2 \\
0 & \leq & -x_a & + & x_b & + & x_c & + & x_d & \leq & 2 
\end{array}
$$
\end{ex}

%
%

We now explain how we can use the cut polytope
to gain information about bounds on individual cell
entries in the table.
Let $Gr_{m,n}$ denote the $m \times n$ grid graph with vertices $v_{1,1}, v_{1,2}, \ldots, v_{m,n}$ and
let $\hat{Gr}_{m,n}$ denote the suspension of $Gr_{m,n}$
obtained from $Gr_{m,n}$ by adding a single vertex
$w$ that is connected to all the other vertices.

\begin{prop}\label{prop:cutising}
Consider the graph $\hat{Gr}_{m,n}$ with associated cut
polytope $CUT^\Box (\hat{Gr}_{m,n}) \subseteq \R^{|E_{1}| + |E_{2}|}$,
where $E_{1}$ denotes the set of edges connected to the vertex $w$ and
$E_{2}$ denotes the set of edges in $Gr_{m,n}$.
Let 
\begin{eqnarray*}
\Sigma' &   =  &   CUT^\Box (\hat{Gr}_{m,n}) \, \,  \bigcap \, \,  \{0,1\}^{|E_{1}| + |E_{2}|}  \\
        &      &   \bigcap  \left\{\mathbf x \in \R^{|E_{1}| + |E_{2}|} \, : \,  \sum\limits_{e\in E_1}  x_e = T_1 \, \, \,  \sum\limits_{e\in E_2}  x_e = T_2 \right\}. \\
\end{eqnarray*}
Let $\pi:  \R^{|E_{1}| + |E_{2}|}  \rightarrow \R^{|E_{1}|}$ be the coordinate
projection onto the first block of coordinates.  Then $\pi(\Sigma')  = \Sigma$.
\end{prop}

\begin{proof}
Since $\pi$ is the coordinate projection, for all $ (y_1, y_2, \ldots, y_{|E_{1}| + |E_{2}|}) \in \R^{|E_{1}| + |E_{2}|}$, $\pi (y_1, \ldots, y_{|E_{1}| + |E_{2}|}) = (y_{1}, \ldots, y_{|E_{1}| })$. For all $\bf x \in \Sigma'$, $\bf x$ is a 0-1 vector in the 0-1 polytope $CUT^\Box (\hat{Gr}_{m,n})$, so there exists a partition of $\hat{Gr}_{m,n}$, $A \mid B$, such that $\mathbf {x} = \delta_{A \mid B}$. Now assume that $w \in A$, and use the partition restricted on ${Gr}_{m,n}$ to define a $m\times n$ table: $x_{i,j} = 1$ if $v_{i,j} \in B$, and 0 if $v_{i,j} \in A$. Then this table $x$ will be identical to $\pi(\bf x)$. Because $\sum\limits_{e\in E_1}  x_e = |B| = T_1$ implies $T_1(x) =T_1$, and $ \sum\limits_{e\in E_2}  x_e = T_2$ implies $T_2(x) = T_2$, we have $\pi(\mathbf{x}) = x \in \Sigma$. $\pi(\Sigma')  \subseteq \Sigma$.

Similarly, given a table $x\in \Sigma$, consider a partition of $\hat{Gr}_{m,n}$, $A \mid B$: $A = \{v_{i,j}: x_{i,j} = 0\}\cup \{w\}$, $B= \{v_{i,j}: x_{i,j} = 1\}$. Let $\mathbf x$ be the resulting cut semimetric. Because $T_1 = T_1(x) = |B| = \sum\limits_{e\in E_1}  x_e$ and $T_2 = T_2(x) = \sum\limits_{e\in E_2}  x_e$, we have $\mathbf x \in \Sigma'$. So $x \in \pi(\Sigma')$. $\Sigma \subseteq \pi(\Sigma')$.
\end{proof}

We apply Proposition \ref{prop:cutising} to
construct linear programs for approximating upper and 
lower bounds on array entries in the Ising model, given 
constrains on certain entries already being filled in the
table.  

To obtain the bounds for a specific cell in the table, we need to optimize a proper target function over $\Sigma$. But since $\Sigma$ contains exponentially many tables, the brute-force search over $\Sigma$ will be very time consuming and not realistic in practice. In fact, this problem is the optimization version of the max-cut problem, which is known to be NP-hard \cite[\S 4.1 and \S 4.4]{deza1997}.
However, Proposition \ref{prop:cutising} implies that we can optimize over $\Sigma'$ instead of $\Sigma$, and the optimal solution over $\Sigma$, i.e. the table in which the value of the specific cell equals to the lower or upper bound, can be uniquely determined by the optimal solution over $\Sigma'$. The advantage of doing this is that the quadratic constraint $T_2(x) = T_2$, which is also the only non-linear constraint in $\Sigma$, is transformed to a linear constraint in $\Sigma'$. Notice that $\Sigma'$ is the set of all lattice points in the polytope defined by intersecting the $CUT^\Box (\hat{Gr}_{m,n})$ with two hyperplanes: one for $T_1$ and one for $T_2$. Therefore, to compute the exact bound of the specific cell, we should solve an integer programming (IP) problem subject to all inequalities that define $CUT^\Box (\hat{Gr}_{m,n})$ and two equations for $T_1$ and $T_2$.

Our difficulty in realizing this idea is that we don't know all inequalities that define $CUT^\Box (\hat{Gr}_{m,n})$, so we define our polytope only using part of them: inequalities for triangles including $w$ and $v_{i_1,j_1}$, $v_{i_2, j_2} \in {Gr}_{m,n}$ such that $e = v_{i_1,j_1} v_{i_2, j_2} \in E_2$ (see Example \ref{ex:tri}); inequalities for squares including four vertices that form a smallest square in $ {Gr}_{m,n}$ (see Example \ref{ex:sq}); equations for $T_1$ and $T_2$ (see the definition of $\Sigma'$ in Proposition \ref{prop:cutising}); and equations for known edges if some cells in the table are already sampled in previous steps.  Then by
using a proper target vector, we can obtain the exact lower and upper bounds of a specific cell by solving IP problems, or approximate them by solving linear programming (LP) problems.

\begin{rem}
Sometimes the bounds computed in the above steps are not exact. There are two reasons. First, we only include some but not all inequalities that define the cut polytope, i.e. the feasible region we use in the IP / LP problems may be larger than the it should be. Second, Proposition \ref{prop:cutising} suggests that we should solve IP problems to get a 0-1 vector. But solving IP problems is usually much slower than solving LP problems, so when we deal with large table, we choose to solve LP problems but not IP problems, and this means we can not guarantee that the resulting bounds correspond to a lattice point. In a word, we can only estimate the bounds but not computing them exactly, so we will still have rejections in our procedure.
\end{rem}

\begin{rem}
Notice that for an originally $m\times n$ table, we need to solve a IP/LP problem with $3mn + 3(m+n)+4$ variables, $36mn+26(m+n)+16$ inequalities and more than $4(m+n)+8$ equalities to compute the bounds of one single cell. 
For large tables, it is both not realistic and not necessary to compute bounds for all cells in the table, even with LP problems. Therefore in our procedure, we only compute bounds when the number of unknown cells is small and also the ratio of the number of edges still need to be added to the number of 1's still need to be added is relatively small.
\end{rem}


\section{Computational results}\label{sec:comp}

\subsection{Simulation results under Ising models}\label{subsec:simu1}

We use the software package {\tt R} \cite{Rproj} in our simulation study.
To measure the accuracy in the estimated p-values, we use the coefficient of variation ($cv^2$) suggested by \cite{chen2005}:
\[
cv^2 = \frac{var_q\{p({ X})/q({ X})\}}{\E^2_q\{p({ X})/q({ X})\}} = \frac{var_q\{1/q({ X})\}}{\E^2_q\{1/q({ X})\}}.
\]
The value of $cv^2$ can be interpreted as the chi-square distance between the 
two distributions $p(\cdot)$ and $q(\cdot)$, which means the smaller it is, the closer the two
distributions are.
In the mean time, we can also measure the efficiency of the sampling method by using the effective sample size (ESS) $ESS = \frac{N}{1+cv^2}$.

For the following examples, we first sample an ``observed table",
from an underlying true distribution.  The underlying true distribution
is taken to be the Ising model in some instances, and more complicated
models in others.  The observed table is sampled using a Gibbs sample.
Then we estimate $p_1$ and $p_2$ via the SIS procedure outlined in the
previous section. 
When the observed tables are actually generated from the Ising model, we expect non-significant p-values in these examples.

The algorithm for the Gibbs sampler is given in the following steps:
\begin{enumerate}
\item
Fix $\alpha$, $\beta$, and the initial table $X^0$. Let $k=0$.
\item
For $i=1,\ldots, mn$, update $X^{k+1}_i \sim$ Bernoulli distribution with\\
$P(X^{k+1}_i=1 \mid X^{k+1}_1, \ldots X^{k+1}_{i-1}, X^k_{i+1}, \ldots X^k_{mn}) = \frac{e^{c_i^k}}{e^{c_i^k}+1}$, where $c_i^{k+1} = \alpha + \beta |j: i\sim j|-2\beta \sum\limits_{j: j\sim i} X^k_j$ (or $X^{k+1}_j$ if $j<i$). The resulting table is $X^{k+1}$.
\item
repeat step 2 and take $X^{1001}$ as our sample.
\end{enumerate}

\begin{ex}
$10\times 10$, $\alpha=-2$, $\beta=0.1$, $N=5000$, $\delta$: acceptance rate. Computational time: 28 sec - 40 sec.
\begin{center}
\begin{tabular}{cccrrrr}
\toprule[1.2pt] %
$T_1$ & $T_2$ & $u(x_0)$ & $p_1$ & $p_2$ & $\delta$ & ESS \\\hline
$11$ & $40$ & 1 & 0.4712 & 0.8083 & $94.9\%$ & 235.6 \\
$17$ & $58$ & 5 & 0.0505 & 0.2854 & $90.8\%$ & 18.5  \\
$13$ & $48$ & 3 & 0.2373 & 0.4657 & $93.0\%$ & 129.2 \\
$18$ & $60$ & 4 & 0.3367 & 0.5054 & $90.5\%$ & 14.9 \\
$20$ & $72$ & 9 & 0.0374 & 0.0883 & $88.9\%$ & 127.4 \\\bottomrule[1.2pt]
\end{tabular}
\end{center}
\end{ex}

\begin{ex}
$20\times 20$, $\alpha=-3$, $\beta=0.1$, $N=5000$, $\delta$: acceptance rate.
Computational time: 102 sec - 124 sec.
\begin{center}
\begin{tabular}{cccrrrr}
\toprule[1.2pt] %
$T_1$ & $T_2$ & $u(x_0)$ & $p_1$ & $p_2$ & $\delta$ & ESS \\\hline
$28$ & $108$ & 2 & 0.7800 & 0.9078 & $96.4\%$ & 4.13 \\
$24$ & $92$   & 2 & 0.4845 & 0.6224 & $96.7\%$ & 14.03  \\
$27$ & $100$ & 1 & 0.8018 & 0.9281 & $95.5\%$ & 16.09 \\
$21$ & $82$   & 2 & 0.3093 & 0.4888 & $97.1\%$ & 47.67 \\
$28$ & $110$ & 6 & 0.0355 & 0.0738 & $95.9\%$ & 27.57 \\\bottomrule[1.2pt]
\end{tabular}
\end{center}
\end{ex}

\begin{ex}
$50\times 50$, $\alpha=-3.6$, $\beta=0.1$, $N=5000$, $\delta$: acceptance rate. Computational time: 647 sec - 661 sec.
\begin{center}
\begin{tabular}{cccrrrr}
\toprule[1.2pt] %
$T_1$ & $T_2$ & $u(x_0)$ & $p_1$ & $p_2$ & $\delta$ & ESS \\\hline
$79$   & $304$   & 2  & 0.6138    & 0.9566 & $98.9\%$ & 7.92 \\
$103$ & $396$   & 4  & 0.9906    & 0.9998 & $98.0\%$ & 5.92  \\
$95$   & $372$   & 8  & 0.1341    & 0.6140 & $98.0\%$ & 3.79 \\
$101$ & $392$   & 11 & $<$0.0001 & 0.0006 & $98.3\%$ & 1.01 \\
$97$   & $370$   & 11 & 0.0189  & 0.0204 & $98.5\%$ & 1.77 \\\bottomrule[1.2pt]
\end{tabular}
\end{center}
\end{ex}

\begin{ex}
$100\times 100$, $\alpha=-2.4$, $\beta=0.1$, $N=5000$, $\delta$: acceptance rate.
Computational time: $\approx$ 2800 sec (47 mins). $p_1 \approx p_2 \approx 1$.
\begin{center}
\begin{tabular}{cccrr}
\toprule[1.2pt] %
$T_1$ & $T_2$ & $u(x_0)$ & $\delta$ & ESS \\\hline
$1127$ & $4096$ & 212 & $97.7\%$ & 1.79 \\
$1121$ & $4064$ & 194 & $97.8\%$ & 2.66  \\
$1104$ & $3962$ & 179 & $97.8\%$ & 1.19 \\
$1149$ & $4112$ & 213 & $97.5\%$ & 2.81 \\
$1117$ & $4030$ & 203 & $97.9\%$ & 1.05 \\\bottomrule[1.2pt]
\end{tabular}
\end{center}
\end{ex}

\subsection{Simulation results under the second-order autologistic regression models}\label{subsec:simu2}

The simulations in this subsection is very similar with Subsection \ref{subsec:simu1}, but the observed tables are generated from the second order autologistic regression models \cite{he2003}. In \cite{he2003}, He et al defines the model as following:
\[
P(X_{i,j}=1 \mid \mbox{all other cells} ) = \frac{e^{c_{i,j}}}{1+ e^{c_{i,j}}},
\]
where $c_{i,j} = \beta_0 + \beta_1 T_{i,j}^{(1)}(X) + \beta_2T_{i,j}^{(2)}(X) + \beta_3 T_{i,j}^{(3)}(X) + \beta_4 T_{i,j}^{(4)}(X) $, $T_{i,j}^{(1)}(X) = X_{i,j-1}+X_{i,j+1}$, $T_{i,j}^{(2)}(X) = X_{i-1,j}+X_{i+1,j}$, $T_{i,j}^{(3)}(X) = X_{i-1,j-1}+X_{i+1,j+1}$, and $T_{i,j}^{(4)}(X) = X_{i-1,j+1}+X_{i+1,j-1}$. 
Similarly to Subsection \ref{subsec:simu1} we can also use Gibbs sampler to generate 
``observed table" from this model and compute its
 $T_1$ and $T_2$. Notice that the Ising models are special
  cases of the second-order autologistic regression models when 
  $\beta_1 = \beta_2$, $\beta_3 = \beta_4 = 0$. 
  Therefore, in the following examples, we expect significant p-values, and theoretically it should be harder to distinguish this model 
  from the Ising model when $\beta_3$ and $\beta_4$ are both small.
  
  To compare with the test statistic $u(\cdot)$ suggested by Besag and Clifford \cite{besag1989}, we will also include a new test statistic suggested by Martin del Campo et al \cite{abraham}: the number of $2\times 2$ windows with adjacent pairs, i.e. $u'(x)$ is the number of $2\times 2$ windows of form
  $\begin{array}{|c|c|}                                                                                       
\hline                                                                                                       
0&0\\                                                                                                      
\hline                                                                                                       
1&1 \\                                                                                                     
\hline                                                                                                       
\end{array}$ in $x$. The corresponding p-values, $p'_1$ and $p'_2$, are similarly defined to $p_1$ and $p_2$.

\begin{ex}
$\beta_0=-2$, $\beta_1=0.2$, $\beta_2=-0.2$, $\beta_3=0.2$, $\beta_4=-0.2$ and $N=5000$. The acceptance rates $\delta$, ESS and computational times won't be listed since they are similar with those in Subsection \ref{subsec:simu1}. 
\begin{center}
\begin{tabular}{ccccrrrrr}
\toprule[1.2pt] %
$m\times n$ & $T_1$ & $T_2$ & $u(x_0)$ & $p_1$ & $p_2$ & $u'(x_0)$ & $p'_1$ & $p'_2$ \\\hline
$10 \times 10$ & $12$  & $42$   & 2  & 0.1887    & 0.4207 & 5 & 0.2301 & 0.4421 \\
$10 \times 10$ & $17$ & $50$   & 3  & 0.0512    & 0.4447 & 8 & 0.3824 & 0.6721 \\
$20 \times 20$ & $53$   & $182$   & 12  & 0.1563   & 0.1956 & 26 & 0.0044 & 0.0297 \\
$20 \times 20$ & $35$   & $130$   & 4  & 0.2424    & 0.6956 & 8 & 0.4682 & 0.9245 \\
$20 \times 20$ & $39$   & $138$   & 2  & 0.9489    & 0.9973 & 14 & 0.2472 & 0.3350 \\
$50 \times 50$ & $326$ & $1154$   & 73 & 0.4499 & 0.9798 & 121 & $<$0.0001 & $<$0.0001 \\
$50 \times 50$ & $298$   & $1044$   & 46 & 0.9884 & 0.9884 & 116 & 0.0008 & 0.0014 \\\bottomrule[1.2pt]
\end{tabular}
\end{center}
The test statistic $u(x)$ failed to reject in all tables, while $u'(x)$ was able to reject some of them.
\end{ex}

\subsection{Real data analysis}\label{subsec:real}
In \cite{besag1978}, Besag published his endive data of size $m=14$ and $n=179$ \cite{friel2004}. In this spread of $14\times 179$ lattices of plants, diseased plants were recorded as 1's and others were 0's. A longitudinal dataset with time-points 0, 4, 8 and 12 weeks was available, and it is well known that Ising model gives very poor fit to the data at 12 weeks \cite{friel2004}.
We applied our method to this data with the following statistics: $T_1 = 385$, $T_2 = 1108$, $u(x_0)=61$. 
$10000$ samples were generated in 1689 sec (28 min). The acceptance rate was $96.3$ and the p-value was between $8.33e$-$5$ and $0.0010$ (ESS was $1.07$). This means we agree with the former research and also reject the Ising model for this data.

The largest real dataset example we have computed is a $800 \times 800$ sparse example with ratio $T_2/T_1$ close to 4. We obtained $5000$ samples in 56 hours, the acceptance rate was $99.6\%$.


\section{Discussion and future work}\label{sec:diss}

We have introduced an efficient SIS procedure which samples two-way 0-1 tables from Ising models. This procedure can be used to carry out conditional exact tests and volume tests. We have also described how to use cut polytopes to reduce rejections in the procedure.
Computational results show that our method performs very well when the table is sparse and 1's are spread out (i.e. the ratio $T_2/T_1$ is relatively large): the acceptance rates are very high, and the computational times are very short.
The method is still fast for the opposite situations, however, the acceptance rates will be much lower. One straightforward way is that we may choose to compute bounds for more cells in the table, but at the price of worse time complexity.

We also observe that the $cv^2$ in our examples are large when the size of table is large, and this can hardly be improved by increasing the sample size. The reason is that our proposal distribution $q(\cdot)$ is not close enough to the true distribution $p(\cdot)$.  A major research problem is to try to   find a better approach of choosing proper marginal distributions $q(x_{i} \mid x_{i-1}, \ldots, x_1)$ (see Section \ref{sec:sis}) so that $p(\cdot)$ and $q(\cdot)$ are closer.

It is also important to choose good test statistics. Simulations in Subsections \ref{subsec:simu1} and \ref{subsec:simu2} suggest that Besag and Clifford's test statistic $u(\cdot)$ gives a small type I error,  but also has a poor power to detect the second-order autologistic regression models, while the other test statistic $u'(\cdot)$ suggested by \cite{abraham} has better capability of detecting these models. We should use different test statistics for different types of alternative models, and a good test statistic should be a feature that has a large probability in the alternative models regardless of the choice of parameters, and has a small probability in Ising models regardless of the choice of parameters. 
For example, Besag's test succeeded in his endive data, but failed for the second-order autologistic regression models, which suggest that $u(x)$ has a larger probability of appearance in the model that the endive data came from than in Ising models and has a similar or smaller probability in the second-order autologistic regression models. And we can similarly interpret that Martin del Campo et al's test \cite{abraham} performs better than Besag's test in our simulations (see Subsection \ref{subsec:simu2}), but not as good as for their alternative models which are special cases of the second-order autologistic regression models when $\beta_1=\beta_2$ and $\beta_3=\beta_4$. 

It is also possible to generalize our method to higher dimensional Ising models, but finding proper marginal distributions in higher dimensional cases will be even bigger challenge than finding them in 2-dimensional case.

\section{Acknowledgment}
We thank Caroline Uhler and Abraham Martin del Campo for useful conversations.
Jing Xi was partially supported by the David and Lucille Packard Foundation.
Seth Sullivant was partially supported by the David and Lucille Packard 
Foundation and the US National Science Foundation (DMS 0954865).

\bibliographystyle{plain}
\bibliography{ising}

\end{document}